\documentclass[submission,copyright,creativecommons]{eptcs}
\providecommand{\event}{MFPS 2021} 
\usepackage{breakurl}             

\usepackage{graphicx}
\usepackage[utf8]{inputenc}
\usepackage{amsfonts}
\usepackage{amsmath}
\usepackage{amsbsy}
\usepackage{mathrsfs}
\usepackage{amsmath}
\usepackage{amsthm}
\usepackage{url}
\usepackage{ebproof}
\usepackage{proof}
\usepackage{color}
\usepackage{graphicx}
\usepackage{ebproof}
\usepackage{cmll}
\usepackage{stmaryrd}
\usepackage{wrapfig}
\usepackage{bm} 

\usepackage{cleveref}
\usepackage{thmtools,thm-restate}
\usepackage{mathtools}
\usepackage{xspace}
\usepackage{cancel}
\usepackage[normalem]{ulem}
\usepackage{diagrams}

\newtheorem{definition}{Definition}[section]

\newtheorem{theorem}[definition]{Theorem}

\newtheorem{lemma}[definition]{Lemma}
\newtheorem{proposition}[definition]{Proposition}
\newtheorem{remark}[definition]{Remark}

\newcommand{\myparagraph}[1]{\medskip \noindent {\bf #1}.}
\newcommand{\lamImp}{\lambda_{imp}}
\newcommand{\Dcat}{{\cal D}}
\newcommand{\iso}{\cong}
\newcommand{\getF}[2]{\textit{get}_{#1} (#2)}
\newcommand{\setF}[3]{\textit{set}^\Filt_{#1} (#2, #3)}
\newcommand{\FILTS}[1]{{\tt Filt}\Set{#1}}

\newcommand{\ALG}{\omega\text{-}\textrm{\bf ALG}}

\newcommand{\Then}{\Longrightarrow}

\newcommand{\Set}[1]{\{#1\}}

\newcommand{\unit}[1]{[ #1 ]}
\newcommand{\UnitSub}[1]{\mbox{\it unit}_{#1}}
\newcommand{\Unit}{\mbox{\it unit}\;}

\newcommand{\Bind}{\star}
\newcommand{\get}[2]{\textit{get}_{#1} (#2)}
\newcommand{\set}[3]{\textit{set}_{#1} (#2, #3)}

\newcommand{\Subst}[2]{[ #1 / #2]}

\newcommand{\ValTerm}{\textit{Val}}
\newcommand{\ComTerm}{\textit{Com}}

\newcommand{\LeftUnit}{\textit{Left unit}}
\newcommand{\RightUnit}{\textit{Right unit}}
\newcommand{\Assoc}{\textit{Assoc}}

\newcommand{\Let}[3]{\mbox{\textit{let} $#1 = #2$ \textit{in} $#3$}}

\newcommand{\Var}{\textit{Var}}

\newcommand{\Sem}[1]{\llbracket   #1 \rrbracket}

\newcommand{\SemD}[1]{\Sem{#1}^D}
\newcommand{\SemSD}[1]{\Sem{#1}^{\GS D}}

\makeatletter
\newcommand{\circlearrow}{}
\DeclareRobustCommand{\circlearrow}{%
	\mathrel{\vphantom{\rightarrow}\mathpalette\circle@arrow\relax}%
}
\newcommand{\circle@arrow}[2]{%
	\m@th
	\ooalign{%
		\hidewidth$#1\circ\mkern1mu$\hidewidth\cr
		$#1\longrightarrow$\cr}%
}
\makeatother

\renewcommand{\int}{\mathsf{i}}

\newcommand{\Inter}{\wedge}

\newcommand{\Th}{\textit{Th}}

\newcommand{\tuple}[1]{\langle #1 \rangle}

\newcommand{\der}{\vdash}

\newcommand{\Env}{\mbox{\it Env}}
\newcommand{\metalambda}{%
	\mathop{%
		\rlap{$\lambda$}%
		\mkern2mu
		\raisebox{.275ex}{$\lambda$}%
	}%
}

\newcommand{\Filt}{{\cal F}}
\newcommand{\FILT}[1]{{\tt Filt}\Set{#1}}

\newcommand{\UnitF}{\textit{unit}^\Filt}

\newcommand{\CatC}{\mathcal{C}}

\newcommand{\CatDom}{\textbf{Dom}}

\newcommand{\Label}{\mathbb{L}}

\newcommand{\Store}{S}

\newcommand{\GS}{\mathbb{S}}
\newcommand{\GSeta}{\eta^{\GS}}
\newcommand{\GSext}[1]{#1^{\GS}}

\usepackage{xcolor}

\newcommand{\RED}[1]{{\color{red}{#1}}}

\newcommand{\violet}[1]{{\color{violet}{#1}}}

\newcommand{\dom}[1]{\textit{dom}( #1 )}

\renewcommand{\to}{\xrightarrow{}}

\newcommand{\ie}{\emph{i.e.}\xspace}

\makeatletter
\newcommand{\uset}[3][0ex]{%
	\mathrel{\mathop{#3}\limits_{
			\vbox to#1{\kern-6\ex@
				\hbox{$\scriptstyle#2$}\vss}}}}
\makeatother

\RequirePackage{ifthen}

\newcommand{\version}{0}
\newcommand{\commenti}{0} 
\newcommand{\condinc}[2]{\ifthenelse{\equal{\commenti}{0}}{#1}{**\violet{#2}} }
\newcommand{\SLV}[2]{\ifthenelse{\equal{\version}{0}}{#1}{ \RED{#2}}}

\newcommand{\storeMap}{\varsigma}

\newcommand{\finMap}{\varsigma}

\newcommand{\omegaD}{\omega_D}
\newcommand{\omegaS}{\omega_{\Store}}
\newcommand{\omegaC}{\omega_{C}}
\newcommand{\omegaSD}{\omega_{\GS D}}

\newcommand{\leqD}{\leq_D}
\newcommand{\leqS}{\leq_S}
\newcommand{\leqC}{\leq_{C}}
\newcommand{\leqSD}{\leq_{\GS D}}

\newcommand{\Lang}{\mathcal{L}}

\newcommand{\omegaR}{(\omega)}
\newcommand{\interR}{(\Inter)}
\newcommand{\leqR}{(\leq)}
\newcommand{\varR}{(\textit{var})}
\newcommand{\lambdaR}{(\lambda)}
\newcommand{\unitR}{(\textit{unit})}
\newcommand{\bindR}{(\Bind)}
\newcommand{\getR}{(\textit{get})}
\newcommand{\setR}{(\textit{set})}

\newdimen\proofrulebreadth \proofrulebreadth=.05em
\newdimen\proofdotseparation \proofdotseparation=1.25ex
\newdimen\proofrulebaseline \proofrulebaseline=2ex
\newcount\proofdotnumber \proofdotnumber=3
\let\then\relax
\def\hfi{\hskip0pt plus.0001fil}
\mathchardef\squigto="3A3B
\newif\ifinsideprooftree\insideprooftreefalse
\newif\ifonleftofproofrule\onleftofproofrulefalse
\newif\ifproofdots\proofdotsfalse
\newif\ifdoubleproof\doubleprooffalse
\let\wereinproofbit\relax
\newdimen\shortenproofleft
\newdimen\shortenproofright
\newdimen\proofbelowshift
\newbox\proofabove
\newbox\proofbelow
\newbox\proofrulename
\def\shiftproofbelow{\let\next\relax\afterassignment\setshiftproofbelow\dimen0 }
\def\shiftproofbelowneg{\def\next{\multiply\dimen0 by-1 }%
\afterassignment\setshiftproofbelow\dimen0 }
\def\setshiftproofbelow{\next\proofbelowshift=\dimen0 }
\def\setproofrulebreadth{\proofrulebreadth}

\def\prooftree{
\ifnum  \lastpenalty=1
\then   \unpenalty
\else   \onleftofproofrulefalse
\fi
\ifonleftofproofrule
\else   \ifinsideprooftree
        \then   \hskip.5em plus1fil
        \fi
\fi
\bgroup
\setbox\proofbelow=\hbox{}\setbox\proofrulename=\hbox{}%
\let\justifies\proofover\let\leadsto\proofoverdots\let\Justifies\proofoverdbl
\let\using\proofusing\let\[\prooftree
\ifinsideprooftree\let\]\endprooftree\fi
\proofdotsfalse\doubleprooffalse
\let\thickness\setproofrulebreadth
\let\shiftright\shiftproofbelow \let\shift\shiftproofbelow
\let\shiftleft\shiftproofbelowneg
\let\ifwasinsideprooftree\ifinsideprooftree
\insideprooftreetrue
\setbox\proofabove=\hbox\bgroup$\displaystyle 
\let\wereinproofbit\prooftree
\shortenproofleft=0pt \shortenproofright=0pt \proofbelowshift=0pt
\onleftofproofruletrue\penalty1
}

\def\eproofbit{
\ifx    \wereinproofbit\prooftree
\then   \ifcase \lastpenalty
        \then   \shortenproofright=0pt  
        \or     \unpenalty\hfil         
        \or     \unpenalty\unskip       
        \else   \shortenproofright=0pt  
        \fi
\fi
\global\dimen0=\shortenproofleft
\global\dimen1=\shortenproofright
\global\dimen2=\proofrulebreadth
\global\dimen3=\proofbelowshift
\global\dimen4=\proofdotseparation
\global\count255=\proofdotnumber
$\egroup  
\shortenproofleft=\dimen0
\shortenproofright=\dimen1
\proofrulebreadth=\dimen2
\proofbelowshift=\dimen3
\proofdotseparation=\dimen4
\proofdotnumber=\count255
}

\def\proofover{
\eproofbit 
\setbox\proofbelow=\hbox\bgroup 
\let\wereinproofbit\proofover
$\displaystyle
}%
\def\proofoverdbl{
\eproofbit 
\doubleprooftrue
\setbox\proofbelow=\hbox\bgroup 
\let\wereinproofbit\proofoverdbl
$\displaystyle
}%
\def\proofoverdots{
\eproofbit 
\proofdotstrue
\setbox\proofbelow=\hbox\bgroup 
\let\wereinproofbit\proofoverdots
$\displaystyle
}%
\def\proofusing{
\eproofbit 
\setbox\proofrulename=\hbox\bgroup 
\let\wereinproofbit\proofusing
\kern0.3em$
}

\def\endprooftree{
\eproofbit 
  \dimen5 =0pt
\dimen0=\wd\proofabove \advance\dimen0-\shortenproofleft
\advance\dimen0-\shortenproofright
\dimen1=.5\dimen0 \advance\dimen1-.5\wd\proofbelow
\dimen4=\dimen1
\advance\dimen1\proofbelowshift \advance\dimen4-\proofbelowshift
\ifdim  \dimen1<0pt
\then   \advance\shortenproofleft\dimen1
        \advance\dimen0-\dimen1
        \dimen1=0pt
        \ifdim  \shortenproofleft<0pt
        \then   \setbox\proofabove=\hbox{%
                        \kern-\shortenproofleft\unhbox\proofabove}%
                \shortenproofleft=0pt
        \fi
\fi
\ifdim  \dimen4<0pt
\then   \advance\shortenproofright\dimen4
        \advance\dimen0-\dimen4
        \dimen4=0pt
\fi
\ifdim  \shortenproofright<\wd\proofrulename
\then   \shortenproofright=\wd\proofrulename
\fi
\dimen2=\shortenproofleft \advance\dimen2 by\dimen1
\dimen3=\shortenproofright\advance\dimen3 by\dimen4
\ifproofdots
\then
        \dimen6=\shortenproofleft \advance\dimen6 .5\dimen0
        \setbox1=\vbox to\proofdotseparation{\vss\hbox{$\cdot$}\vss}%
        \setbox0=\hbox{%
                \advance\dimen6-.5\wd1
                \kern\dimen6
                $\vcenter to\proofdotnumber\proofdotseparation
                        {\leaders\box1\vfill}$%
                \unhbox\proofrulename}%
\else   \dimen6=\fontdimen22\the\textfont2 
        \dimen7=\dimen6
        \advance\dimen6by.5\proofrulebreadth
        \advance\dimen7by-.5\proofrulebreadth
        \setbox0=\hbox{%
                \kern\shortenproofleft
                \ifdoubleproof
                \then   \hbox to\dimen0{%
                        $\mathsurround0pt\mathord=\mkern-6mu%
                        \cleaders\hbox{$\mkern-2mu=\mkern-2mu$}\hfill
                        \mkern-6mu\mathord=$}%
                \else   \vrule height\dimen6 depth-\dimen7 width\dimen0
                \fi
                \unhbox\proofrulename}%
        \ht0=\dimen6 \dp0=-\dimen7
\fi
\let\doll\relax
\ifwasinsideprooftree
\then   \let\VBOX\vbox
\else   \ifmmode\else$\let\doll=$\fi
        \let\VBOX\vcenter
\fi
\VBOX   {\baselineskip\proofrulebaseline \lineskip.2ex
        \expandafter\lineskiplimit\ifproofdots0ex\else-0.6ex\fi
        \hbox   spread\dimen5   {\hfi\unhbox\proofabove\hfi}%
        \hbox{\box0}%
        \hbox   {\kern\dimen2 \box\proofbelow}}\doll%
\global\dimen2=\dimen2
\global\dimen3=\dimen3
\egroup 
\ifonleftofproofrule
\then   \shortenproofleft=\dimen2
\fi
\shortenproofright=\dimen3
\onleftofproofrulefalse
\ifinsideprooftree
\then   \hskip.5em plus 1fil \penalty2
\fi
}

\title{From Semantics to Types: the Case of the Imperative $\lambda$-Calculus}
\author{Ugo de'Liguoro
	\institute{Department of Computer Science\\ Università di Torino\\
		Torino, Italy}
	\email{ugo.deliguoro@unito.it}
	\and
	Riccardo Treglia
	\institute{Department of Computer Science\\ Università di Torino\\
		Torino, Italy}
	\email{\quad riccardo.treglia@unito.it }
}
\def\titlerunning{From Semantics to Types: the Case of the Imperative $\lambda$-Calculus}
\def\authorrunning{U. de'Liguoro \& R. Treglia}
\begin{document}
	\maketitle

\begin{abstract} 
  We propose an intersection type system for an imperative $\lambda$-calculus based
  on a state monad and equipped with algebraic operations to read and write to the store.
  The system is derived by solving a suitable domain equation in the category of $\omega$-algebraic lattices;
  the solution consists of a filter-model generalizing the well known construction for ordinary $\lambda$-calculus. Then
  the type system is obtained out of the term interpretations into the filter-model itself.
  The so obtained type system satisfies the ``type-semantics'' property, and it is sound and complete by construction.
\end{abstract}


\section{Introduction}\label{sec:intro}

Since Strachey and Scott's work in the 60's, $\lambda$-calculus and denotational semantics, together with logic and type theory, have been recognized
as the mathematical foundations of programming languages. Nonetheless, there are aspects of actual programming languages that have shown to be
quite hard to treat, at least with the same elegance as the theory of recursive functions and of algebraic data structures; a prominent case is surely
side-effects.

The introduction of notions of computation as monads by Moggi in \cite{Moggi'91} greatly improved the understanding of ``impure'', namely non functional
features in the semantics of programming languages, by providing a unified framework for treating computational effects: see \cite{Wadler92,Wadler-Monads},
the introductory \cite{BentonHM00}, and the large body of bibliography thereafter; a gentle introduction and references can be found e.g. in \cite{amsdottorato9075}.

Starting with \cite{deLiguoroTreglia20}, we have considered an untyped variant of Moggi's computational $\lambda$-calculus, with two sorts of terms: {\em values} denoting points of
some domain $D$, and {\em computations} denoting points of $TD$, where $T$ is some generic monad and $D \cong D \to TD$. The goal  was
to show how such a calculus can be equipped with an operational semantics and an intersection type system, such that types are invariant under reduction and expansion
of computation terms, and convergent computations are characterized by having non-trivial types in the system.

Our approach was limited to a pure calculus without constants, where the unit and bind operations are axiomatized by the monadic laws from  \cite{Wadler-Monads}.
On the other hand, such operations may
model how morphisms from values to computations compose, but do not tell anything
about how the computational effects are produced. 
In the theory of {\em algebraic effects} \cite{PlotkinP02,PlotkinP03,Power06}, Plotkin and Power have shown under which conditions effect operators
live in the category of algebras of a computational monad, which is isomorphic to the category of models of certain equational specifications, 
namely varieties in the sense of universal algebra \cite{Hyland2007TheCT}.

Toward a general theory of intersection type systems for (untyped) computational $\lambda$-calculi with algebraic effects, we have studied in the draft 
\cite{deLiguoroTreglia2021} the case of the imperative $\lambda$-calculus $\lamImp$, both because the modeling of a functional calculus with (implicitly higher-order) 
side-effects is of interest per se, and because it is a case study to test our claim that an intersection types assignment system can be defined and that it is useful to formalize
reasoning about the calculus itself.


Indeed the type system, which we present in the very last section of this work, has not been invented from scratch; on the contrary, it has been constructed moving from the
domain equation $D = D \to \GS D$, where $\GS$ is (a variant of) the state monad. The method we follow is to solve such equation in the category of $\omega$-algebraic lattices, whose
objects and morphisms can be described via intersection type theories, that are the ``logic'' of such domains in the sense of  \cite{Abramsky'91}: see also
\cite{alesdezahons04}. 
The solution is constructed via the description of compact points of such a $D$, implicitly mirroring the inverse limit construction, but in a conceptually simpler way, consisting
of an inductive definition of types and of their preorders.
Once this has been obtained, subtyping and typing rules are derived in a uniform way, leading to a type assignment system which enjoys the 
``type-semantics'' property, namely that  denotational semantics of terms is fully characterized by their typings in the system. 
Soundness and completeness of the type system now follow straightforwardly.

\myparagraph{Summary and results} In section \ref{sec:calculus} we outline the syntax of the untyped imperative $\lambda$-calculus $\lamImp$, and stress the difference with other calculi in literature. In section \ref{sec:denotational} we deal with solving the domain equation thoroughly. To do this we recall the concept of computational monad, then we consider the state and partiality monad $\GS$. In the same section, we tackle the solution of the domain equation by breaking the circularity that arises out of the definition of the monad $\GS$ itself. We conclude the section by modelling algebraic operators over the monad $\GS$ and formalizing the model of $\lamImp$.
In section \ref{sec:intersection}, we solve the domain equation by inductively constructing type theories that induce a filter model. 

Section \ref{sec:deriving} explains how to derive the type assignment system from the filter-model construction. 
We conclude establishing the type semantics theorem for our calculus. Finally,
section \ref{sec:related} is devoted to the discussion of our results and to related works.

\medskip
We assume familiarity with $\lambda$-calculus and intersection types; a comprehensive reference is \cite{BarendregtDS2013} Part III. 
Some previous knowledge about computational monads and algebraic effects is useful:
see e.g. \cite{BentonHM00} and \cite{amsdottorato9075} chap. 3.

\if false

In the study of such

We move from a $\lambda$-calculus enriched with side effects, which is formalized as a computational $\lambda$-calculus in the sense of Moggi with algebraic operators
\`a la Plotkin-Power. Using a variant of the monad $\GS$ we model the calculus into a solution of the domain equation $D = D \to \GS D$, providing
an interpretation map $\SemD{\cdot}$.

Next we build a "logic" of the calculus  in the sense of domain logic, constructing a filter model $\Filt^D$. Finally, by particularizing the interpretation map to $\Filt^D$, we
recover a type assignment system enjoying the type semantics property by construction, hence sound and complete.

------------------\\
Since Strachey and Scott's work in the 60's, $\lambda$-calculus and denotational semantics, together with logic and type theory, have been recognized
as the mathematical foundations of programming languages. Nonetheless, there are aspects of actual programming languages that have shown to be
quite hard to treat, at least with the same elegance as the theory of recursive functions and of algebraic data structures; a prominent case is surely
side-effects.

In \cite{Moggi'91} Moggi proposed a unified framework to reason about $\lambda$-calculi embodying various kinds of effects, including side-effects,
that has been used by Wadler \cite{Wadler92,Wadler-Monads} to cleanly implement non-functional aspects into Haskell, a purely functional programming language. Moggi's
approach is based on the categorical notion of {\em computational monad}: instead of adding  impure effects to the semantics 
of a pure functional calculus, 
effects are subsumed by the abstract concept of ``notion of computation'' represented by the monad $T$. 
To a domain $A$ of values it is associated the domain $TA$ of computations over $A$,
that is an embellished structure in which $A$ can be merged, and such that any morphism $f$ from $A$ to $TB$ extends by a universal construction 
to a map $f^T$ from $TA$ to $TB$. 

Monadic operations of merging values into computations,  \ie the {\em unit} of the monad $T$, and of extension, 
model how morphisms from values to computations compose, but do not tell anything
about how the computational effects are produced. 
In the theory of {\em algebraic effects} \cite{PlotkinP02,PlotkinP03,Power06}, Plotkin and Power have shown under which conditions effect operators
live in the category of algebras of a computational monad, which is isomorphic to the category of models of certain equational specifications, 
namely varieties in the sense of universal algebra \cite{Hyland2007TheCT}.

In \cite{deLiguoroTreglia20} we have considered an untyped computational $\lambda$-calculus with two sorts of terms: {\em values} denoting points of
some domain $D$, and {\em computations} denoting points of $TD$, where $T$ is some generic monad and $D \cong D \to TD$. The goal  was
to show how such a calculus can be equipped with an operational semantics and an intersection type system, such that types are invariant under reduction and expansion
of computation terms, and convergent computations are characterized by having non-trivial types in the system.

Here, we extend our approach and consider a variant of the {\em state monad} from \cite{Moggi'91} and a calculus with two families of operators,
indexed over a denumerable set of {\em locations}: $\get{\ell}{\lambda x.M}$ reading
the value $V$ associated to the location $\ell$ in the current {\em state}, and binding $x$ to $V$ in $M$; $\set{\ell}{V}{M}$ which modifies the state assigning $V$ to $\ell$,
and then proceeds as $M$. This calculus, with minor notational differences, is called {\em imperative $\lambda$-calculus} in \cite{amsdottorato9075}.

Types and type assignment system are derived from the domain equation defining $D$ and $\GS D$, following the method of domain logic in \cite{Abramsky'91}. 
Type and typing rule definitions are guided along the path of a well understood mathematical method, which indicates both how type syntax and the subtyping relations
are constructed and  how to shape type assignment rules. The result is an extension of Curry style intersection type assignment system: see
\cite{BarendregtDS2013} Part III.

\paragraph{Filter Model.}
As a first step, we construct a filter model  $\Filt_D$ that is isomorphic in a category of domains to $D \to \GS D$, where $\GS$ is the {\em monad of partiality and state} in \cite{LagoGL17}.

\paragraph{Deriving type assignment system.} Not towards completeness, but completeness as a natural consequence of the construction.

\paragraph{Results.}


\medskip
We assume familiarity with $\lambda$-calculus and intersection types; a comprehensive reference is \cite{BarendregtDS2013} Part III. 
Some previous knowledge about computational monads and algebraic effects is useful:
see e.g. \cite{BentonHM00} and \cite{amsdottorato9075} chap. 3. 

\fi


\section{An untyped imperative $\lambda$-calculus: $\lamImp$}
\label{sec:calculus}

Imperative extensions of the $\lambda$-calculus, both typed and type-free, are usually
based on the call-by-value $\lambda$-calculus, and enriched with constructs for reading and writing to
the store. Aiming at exploring the semantics of side effects in computational calculi, where 
``impure'' functions are modeled by pure ones sending values to computations in the sense of 
\cite{Moggi'91}, we consider the calculus introduced in \cite{deLiguoroTreglia20}, to which
we add syntax denoting algebraic effect operations \`{a} la Plotkin and Power \cite{PlotkinP02,PlotkinP03,Power06}
over a suitable state monad. 

Let $\Label = \Set{\ell_0, \ell_1, \ldots}$ be a denumerable set of abstract {\em locations}. Borrowing notation from 
\cite{amsdottorato9075}, chap. 3, we consider denumerably many operator symbols
 $\textit{get}_{\ell}$ and $\textit{set}_{\ell}$, obtaining:

\begin{definition} 
\label{def:syntax}
The syntax of the untyped imperative $\lambda$-calculus $\lamImp$ has two sorts of terms, which are defined by mutual induction as follows:
\[\begin{array}{rrcl@{\hspace{1cm}}l}
\ValTerm: & V,W & ::= & x \mid \lambda x.M \\ 
\ComTerm: & M,N & ::= & \unit{V} \mid M \Bind V  
                   \mid \get{\ell}{\lambda x.M} \mid \set{\ell}{V}{M} & (\ell \in \Label) \\ 
\end{array}\]
\end{definition}

\noindent
Terms are of either sorts $\ValTerm$ and $\ComTerm$, representing values and computations, respectively.
The variable $x$ is bound in $\lambda x.M$ and $\get{\ell}{\lambda x.M}$; 
terms are identified up to renaming of bound variables so that the capture avoiding substitution $M\Subst{V}{x}$ is always well defined;
$FV(M)$ denotes the set of free variables in $M$.
We call $\ValTerm^{\,0}$ the subset of closed $V \in \ValTerm$; similarly for $\ComTerm^0$.

Terms of the shape $\unit{V}$ and $M \Bind V$ are  written $\texttt{return } V$ and $M \texttt{ >>= } V$ in Haskell-like syntax, respectively.
With respect to \cite{deLiguoroTreglia20} we write $\unit{V}$ instead of $\Unit V$.
With respect to the syntax of the ``imperative $\lambda$-calculus'' in \cite{amsdottorato9075}, 
we do not have the {\it let} construct, nor the application $VW$ among values.
These constructs are indeed definable: $\Let{x}{M}{N} \;\equiv\; M \Bind (\lambda x.N)$ and $VW \;\equiv\; \unit{W} \Bind V$ (where $\equiv$ is syntactic identity).
In general, application among computations can be encoded by $MN \equiv M \Bind (\lambda z. \ N \Bind z)$, where $z$ is  fresh.

In a sugared notation from functional programming languages, we could have written location dereferentiation and assignment, respectively, as follows:
\[ \Let{x}{\; !\ell}{M} ~\equiv~ \get{\ell}{\lambda x.M}  \qquad 
 \ell := V; M ~\equiv~ \set{\ell}{V}{M} \]
Observe that we do not consider locations as
values; consequently they cannot be dynamically created like with the \textbf{ref} operator from ML,  nor
is it possible to model aliasing. On the other hand, since the calculus is untyped, ``strong updates'' are allowed (differently than \cite{BentonKBH09}). By this, we mean that when evaluating $\set{\ell}{V}{M}$,
the value $V$ to which the location $\ell$ is updated  bears no relation
to the previous value, say $W$, of $\ell$ in the store: indeed, in our type assignment system $W$ and $V$ may well have completely different types.


\section{Denotational semantics}
\label{sec:denotational}

\renewcommand{\GSext}[1]{#1^\dagger}
\renewcommand{\GSeta}{\eta}

We illustrate a denotational semantics of the calculus in the category of domains.
The model is based on the solution of the domain equation:
\begin{equation}\label{eq:main-domain-eq}
D = [D \to[ S \to (D \times S)_\bot]]
\end{equation}
where $S$ is a suitable space of stores over $D$, and $\GS \,D = [S \to (D \times S)_\bot]$ is a variant of
the state monad in \cite{Moggi'91}, which is called the partiality and state monad in \cite{LagoGL17}.


Before embarking on solving the equation (\ref{eq:main-domain-eq}), let us recall definitions of monads, of $\GS$, and algebraic operators, and fix notations.

%
\myparagraph{The partiality and state monad}
We recall Wadler's type-theoretic definition of monads \cite{Wadler92,Wadler-Monads}, that is at the basis of their successful implementation
in Haskell language, a natural interpretation of the calculus is into a cartesian closed category (ccc), such that
two families of combinators, or a pair of polymorphic operators called the ``unit'' and the ``bind'', exist satisfying the {\em monad laws}.
In what follows, $\CatC$ will be  a {\em concrete ccc}. Examples are the category $\CatDom$ of Scott domains with continuous functions, and
its full subcategory $\ALG$ of algebraic lattices with a countable basis.

\begin{definition}\label{def:Cmonad}
{\bf (Computational Monad \cite{Wadler-Monads} \S 3)}
	Let $\CatC$ be a concrete ccc. A {\em functional computational monad}, henceforth just {\em monad}
	over $\CatC$ is a triple $(T, \Unit\!, \Bind)$ where
	$T$ is a map over the objects of $\CatC$, and $\Unit\!$ and $\Bind$ are families of morphisms
	\[\begin{array}{c@{\hspace{2cm}}c}
		\UnitSub{A}: A \to TA & \Bind_{A,B}: TA \times (TB)^A  \to TB
	\end{array}\]
	such that, writing $\Bind_{A,B}$ as an infix operator and omitting subscripts:
	\[\begin{array}{l@{\hspace{0.5cm}}rcl@{\hspace{0.5cm}}l}
		\LeftUnit: & \Unit a \Bind f & = & f\,a \\
		\RightUnit: & m \Bind \Unit & = & m \\
		\Assoc: & (m \Bind f) \Bind g & = & m \Bind \metalambda d. (f\,d \Bind g)
	\end{array}\]
	where
 $\metalambda d. (f\,d \Bind g)$ is the lambda notation for $d \mapsto f\,d \Bind g$.
\end{definition}

\begin{remark} In case of concrete ccc, Def. \ref{def:Cmonad} coincides with the notion of {\em Kleisli triple} $(T, \eta, \_^\dagger)$ (see e.g. \cite{Moggi'91}, Def. 1.2).
The equivalence is obtained by relating maps $\eta$ and $\_^\dagger$ to $\Unit$ and $\Bind$ operators, as follows:
\[ \UnitSub{X} = \eta_X  \qquad \qquad a \Bind f = f^\dagger (a) \]
The first equality is just a notational variation of the coercion of values into computations. The latter, instead, establishes the connection between the $\Bind$ operator with the map $\_^\dagger$, also called {\em extension operator}: if $f:X \to TY$ then $f^\dagger:TX \to TY$ is the unique map such that $f = f^\dagger \circ \eta_X$.
\end{remark}

Now, let us look closer how monads model effectful computations. We focus on $\CatDom$, 
a full subcategory of pointed Dcpo's with Scott continuous functions as morphisms.
Then $TX$ is the domain, or the type, of computations with values in $X$. In general $TX$ has a richer structure
than $X$ itself, modeling partial computations, exceptions, non-determinism etcetera, including side effects. 
The mapping $\UnitSub{X}:X \to TX$ is interpreted as the trivial computation  $\UnitSub{X}(x)$ just returning the value $x$, and it is an embedding if
$T$ satisfies the requirement that all the $\UnitSub{X}$ are monos. 
Now a function $f:X \to TY$ models an ``impure'' program $P$ with input in $X$ and output in $Y$
via a pure function returning the computation $f(x) \in TY$. 

To relate equation (\ref{eq:main-domain-eq}) to the monad $\GS$ we have to be more precise about the domain $\Store$. We consider $F X =  X^{\Label}$ the domain of {\em stores} over $X$ with locations in $\Label$, namely the set $X^{\Label}$, ranged over by $\varsigma$, of maps from locations to points of $X$; $FX$ is a domain with the order inherited from $X$. The map $FX$ obviously extends to a locally continuous functor by setting $F(g) = g^{\Label} = g \circ \_:X^\Label \to Y^\Label$ for any $g:X \to Y$.

We are now in place to properly define the partiality and state monad $\GS$ in terms of Def. \ref{def:Cmonad}.

\begin{definition}\label{def:statemon}
{\bf (Partiality and state monad)}
	Given the domain $S = X^\Label$ representing a notion of state, we define the partiality and state monad $(\GS, \Unit, \Bind)$, as the mapping
	\[\GS\,X = [S \to (X \times S)_\bot]\]
	where $(X \times S)_\bot$ is the lifting of the cartesian product $X \times S$, equipped with two (families of) operators $\Unit$ and $\Bind$ defined as follows:
	
	\[\begin{array}{l@{\hspace{0.5cm}}rcl@{\hspace{0.5cm}}l}
		\Unit x::= \metalambda \finMap. (x,\finMap)\\
		(c\Bind f)\finMap =	f^\dagger(c)(\finMap) ::= \left \{
		\begin{array}{lll}
			f \,(x) (\finMap')\, & \mbox{if $c(\finMap)= (x, \finMap')\neq \bot$} \\
			\bot                & \mbox{if $c(\finMap) =  \bot$}
		\end{array} \right.
	\end{array}\]
where we omit subscripts.
\end{definition}


\begin{remark}
A function $f:X \to \GS\,Y$ has type $X \to S \to (Y \times S)_\bot$, which is isomorphic to $(X \times S) \to (Y \times S)_\bot$;
if $f $ is the interpretation of an ``imperative program'' $P$, and it is the currying of $f'$, then we expect that  
$f\,x\,\finMap = f'(x,\finMap) = (y, \finMap')$ if $y$ and $\finMap'$ are respectively the value and the final store obtained from
the evaluation of $P(x)$ starting in $\finMap$, if it terminates; $f\,x\,\finMap = f'(x,\finMap) = \bot$ otherwise. Clearly $f'$ is the familiar
interpretation of the imperative program $P$ as a state transformation mapping.
\end{remark}

\myparagraph{A domain equation} Evidently, the domain $FX$ depends on $X$ itself; in contrast, while defining $\GS\,X$ the
domain $S$ has to be fixed, since otherwise the definition of the $\Bind$ operator does not make sense, and $\GS$ 
is not even a functor. A solution would be to take $S = FD$, where $D \cong [ D \to \GS\,D]$, but this is clearly circular.

To break the circularity, we define the mixed-variant bi-functor $G: \CatDom^\textit{op} \times \CatDom \to \CatDom$ by
\[G(X,Y) = [FX \to (Y \times FY)_\bot] \]
whose action on morphisms is illustrated by the diagram:
\[
\begin{diagram}[width=4em]
FX'                            & \rTo^{Ff} & FX\\
\dDashto^{G(f,g)(\alpha)} & & \dTo_{\alpha} \\
 (Y' \times FY')_\bot & \lTo_{\quad (g \times Fg)_\bot \quad} & (Y \times FY)_\bot\\
\end{diagram}
\]
where $f:X' \to X$, $g:Y \to Y'$ and $\alpha \in G(X,Y)$. Now it is routine to prove that
$G$ is locally continuous so that, by the inverse limit technique, 
we can find the initial solution to the domain equation (which is the same as \Cref{eq:main-domain-eq}):
\begin{equation}\label{eq:domain-equation}
D = [D \to G(D,D)]
\end{equation}
In summary we have:

\begin{theorem}\label{thm:domain-equation}
There exists a domain $D$ such that the state monad $\GS_D$ with state domain 
$S = [\Label \to D]$ is a solution in $\CatDom$ to the domain equation:
\[D = [D \to \GS_D\,D]\]
Moreover, it is initial among all solutions to such equation.
\end{theorem}

\begin{proof} Take $D$ to be the (initial) solution to \Cref{eq:domain-equation};
now if $S = FD = D^\Label$ then $\GS_D\,D = G(D,D)$.
\end{proof}

\myparagraph{Algebraic operations over $\GS$}
Monads are about composition of morphisms of some special kind.
However, thinking of $f:X \to \GS X$ as the meaning of a program
with side effects doesn't tell anything about side effects themselves, that are produced by reading and writing
values from and to stores in $S$. 

To model effects associated to a monad, Plotkin and Power have proposed in \cite{PlotkinP02,PlotkinP03,Power06}
a theory of {\em algebraic operations}. A gentle introduction to the theory can be found in 
\cite{amsdottorato9075}, chap. 3, from which we borrow the notation.

Suppose that $T$ is a monad over a category $\CatC$ with
terminal object $\textbf{1}$ and all finite products; then an algebraic operation $\textbf{op}$ with arity $n$ is a family of
morphisms $\textbf{op}_X:(TX)^n \to TX$, where $(TX)^n = TX \times \cdots \times TX$ is the $n$-times product of $TX$ with itself, such that
\begin{equation}\label{eq:n-ary-op}
\textbf{op}_X \circ \Pi_n f^\dagger = f^\dagger \circ \textbf{op}_X
\end{equation}
where $f: X \to TX$ and $\Pi_n f^\dagger = f^\dagger \times \cdots \times f^\dagger: (TX)^n \to (TX)^n$.
In case of the concrete category $\CatDom$, $\textbf{1}$ is a singleton and products are sets of tuples. 
Then $\textbf{op}_X$ is an operation of arity $n$ of an algebra with carrier $TX$; since $f^\dagger: TX \to TX$,
we have that $(\Pi_n f^\dagger)\tuple{x_1, \ldots , x_n} = \tuple{f^\dagger(x_1), \ldots , f^\dagger(x_n)}$, and \Cref{eq:n-ary-op}
reads as:
\[\textbf{op}_X(f^\dagger(x_1), \ldots , f^\dagger(x_n)) = f^\dagger(\textbf{op}_X(x_1, \ldots , x_n))\]
namely the functions $f^\dagger$ are homomorphisms w.r.t. $\textbf{op}_X$.

Algebraic operations $\textbf{op}_X:(TX)^n \to TX$ do suffice in case $T$ is, say, the nondeterminism or the output monad, but cannot
model side effects operations in case of the store monad. This is because read and write operations implement a bidirectional action
of stores to programs and of programs to stores. What is needed instead is the generalized notion of operation proposed in \cite{PlotkinP02} (and further studied in \cite{HylandP06}); such construction,
that is carried out in a suitable enriched category,  can be instantiated to the case of $\CatDom$ as follows:
\[\textbf{op}:  P \times (TX)^A \to TX \cong (TX)^A \to (TX)^P\]
where $P$ is the domain of parameters and $A$ of generalized arities, and the rightmost ``type'' is the interpretation in $\CatDom$ of Def. 1 in \cite{PlotkinP02}. 
For such operations \Cref{eq:n-ary-op} is generalized as follows (see \cite{amsdottorato9075}, Def. 13):
\begin{equation}\label{eq:gen-op}
\textbf{op}(p,k) \Bind f = 
f^\dagger (\textbf{op}(p,k)) = \textbf{op}(p, f^\dagger \circ k)
= \textbf{op}(p, \metalambda x. (k(x) \Bind f))
\end{equation}
where $f: X \to TX$, $p \in P$ and $k:A \to TX$.

\myparagraph{Denotational semantics of terms}
Given the monad 
$(\GS, \eta, \_^\dagger)$ and the domain $D$ from Theorem \ref{thm:domain-equation}, 
we first interpret the constants $\textit{get}_{\ell}$ and $\textit{set}_{\ell}$ as generalized operations over $\GS$. 
By taking $P = \textbf{1}$ and $A = D$ we define:
\[\Sem{\textit{get}_{\ell}}:  \textbf{1} \times (\GS D)^D \to \GS D \simeq (\GS D)^D \to \GS D\quad
\mbox{ by } \quad  \Sem{\textit{get}_{\ell}} \, d \, \storeMap = d\, (\storeMap \,( \ell))\, \storeMap\] 
where $d \in D$ is identified with its image in $D \to \GS\,D$, and $\storeMap \in S = [\Label \to D]$.
On the other hand, taking $P = D$ and $A = \textbf{1}$, we define:
\[\Sem{\textit{set}_{\ell}}: D \times (\GS D)^{\textbf{1}} \to \GS D \simeq D \times \GS D \to \GS D \quad
\mbox{ by }  \Sem{\textit{set}_{\ell}}(d, c)\, \storeMap = c (\storeMap[\ell \mapsto d])\] 
where $c \in \GS D = S \to (D \times S)_\bot$ and $\storeMap[\ell \mapsto d]$ is the store
sending $\ell$ to $d$ and it is equal to $\storeMap$ otherwise.

Then we interpret values from $\ValTerm$ in $D$
and computations from $\ComTerm$ in $\GS\,D$. More precisely we define the maps
$\SemD{\cdot}: \ValTerm \to \Env \to D$ and 
$\SemSD{\cdot}: \ComTerm \to \Env \to \GS\, D$,
where $\Env = \Var \to D$ is the set of environments interpreting term variables in $\Var$:

\begin{definition}\label{def:lambda-imp-model}
\label{def:lam-imp-mod}
A \textbf{$ \lamImp $-model} is a structure $\Dcat = (D, \GS, \SemD{\cdot},\SemSD{\cdot})$ such that:
\begin{enumerate}
	\item $D$ is a domain s.t. $D\cong D\to \GS D$ via $(\Phi,\Psi)$, where $\GS$ is the partiality and state monad of stores over $D$;
	\item for all $e \in \Env$, $V \in \ValTerm$ and $M \in \ComTerm$:
		\[\begin{array}{ll}
			\Sem{x}^D e  =  e (x) &\quad
			\SemD{\lambda x.M} e  = \Psi( \metalambda d \in D. \, \SemSD{M} {e[x  \mapsto d]} )\\ [1mm]
			\SemSD{\, \unit{V} \,} e = \Unit \, (\SemD{V} e) & \quad
			\SemSD{M \Bind V} e  = (\SemSD{M} e)\Bind \Phi (\SemD{V} e )\\ [1mm]
			\SemSD{ \get{\ell}{\lambda x.M}} e =  \Sem{\textit{get}_{\ell}} \, \Phi (\SemD{\lambda x.M} e)  & \quad
			\SemSD{ \set{\ell}{V}{M} } e  = \Sem{\textit{set}_{\ell}}(\SemD{V} e, \SemSD{M}  e) 
		\end{array}\]
\end{enumerate}
\end{definition}
By unravelling definitions and applying to an arbitrary store $\storeMap \in S$, the last two clauses can be written:
\[\begin{array}{rcl}
\SemSD{ \get{\ell}{\lambda x.M}} e \, \storeMap & = & \SemSD{M} (e [x  \mapsto \storeMap(\ell)]) \, \storeMap \\ [1mm]
\SemSD{ \set{\ell}{V}{M} } e \, \storeMap & = & \SemSD{M}  e \, (\storeMap[\ell \mapsto \SemD{V} e])
\end{array}\]

We say that the equation $M = N$ is {\em true} in $\Dcat$, written $\Dcat \models M = N$, if $\SemSD{M} e = \SemSD{N} e$ for all $e \in \Env$.

\begin{proposition}\label{prop:true-eq}
The following equations are true in $\Dcat$:
\begin{enumerate}
\item \label{prop:true-eq-i} $\unit{V} \Bind (\lambda x.M) = M\Subst{V}{x}$
\item \label{prop:true-eq-ii} $M \Bind \lambda x. \unit{x} = M$
\item \label{prop:true-eq-iii} $(L \Bind \lambda x.M) \Bind \lambda y.N = L \Bind \lambda x. (M \Bind \lambda y. N)$
\item \label{prop:true-eq-iv} $\get{\ell}{\lambda x.M} \Bind W = \get{\ell}{\lambda x.(M \Bind W)}$
\item \label{prop:true-eq-v} $\set{\ell}{V}{M} \Bind W = \set{\ell}{V}{M \Bind W}$
\end{enumerate}
where $x \not \in FV(\lambda y.N)$  in (\ref{prop:true-eq-iii}) and $x \not\in FV(W)$ in (\ref{prop:true-eq-iv}).
\end{proposition}

\begin{proof}
By definition and straightforward calculations. For example, to see (\ref{prop:true-eq-iv}), let $\storeMap \in S$ be arbitrary
and $e' = e [x  \mapsto \storeMap(\ell)]$; then, omitting the apices of the interpretation mappings $\Sem{\cdot}$:
\[
\Sem{\get{\ell}{\lambda x.M} \Bind W} e \, \storeMap =
	\GSext{(\Sem{W} e)}(\Sem{M} e') \, \storeMap
\]
On the other hand:
\[
\Sem{\get{\ell}{\lambda x.(M \Bind W)}} e \, \storeMap 
= (\Sem{\lambda x.(M \Bind W)} e) \, \storeMap(\ell) \, \storeMap
= \Sem{M \Bind W} e' \, \storeMap 
= \GSext{(\Sem{W} e')}(\Sem{M} e') \, \storeMap
\]
But $x \not\in FV(W)$ implies $\Sem{W} e' = \Sem{W} e$, and we are done.
\end{proof}

Equations (\ref{prop:true-eq-i})-(\ref{prop:true-eq-iii}) in Prop. \ref{prop:true-eq} are the {\em monadic laws} from Def. \ref{def:Cmonad}, equations (\ref{prop:true-eq-iv})-(\ref{prop:true-eq-v})
are instances of \Cref{eq:gen-op}.


\section{The filter-model construction}
\label{sec:intersection}

In this section we switch to the category $\ALG$ of $\omega$-algebraic lattices. $\ALG$ is the full subcategory of $\CatDom$ of complete lattices 
$(D, \sqsubseteq_D)$ such that the set $K(D)$ of compact points is countable and $d = \bigsqcup \Set{e \in K(D) \mid e \sqsubseteq_D d}$ is a directed sup for all $d \in D$.

Observe that morphisms in $\ALG$ are Scott-continuous maps, hence preserving direct sups, but not necessary arbitrary ones.

\begin{definition}\label{def:theories}
An {\em intersection type theory}, shortly {\em itt}, is a pair $\Th_A = (\Lang_A, \leq_A)$ where $\Lang_A$, the {\em language} of 
$\Th_A$, is a countable set of type expressions closed under $\Inter$, and $\omega_A\in \Lang_A$ is a special constant; $\leq_A$ 
is a pre-order over $\Lang_A$ closed under the following rules:
\[\begin{array}{c@{\hspace{1.5cm}}c@{\hspace{1.5cm}}c@{\hspace{1.5cm}}c@{\hspace{1.5cm}}c}
\alpha \leq_A \omega_A &
\alpha \Inter \beta \leq_A \alpha &
\alpha \Inter \beta \leq_A \beta &
\prooftree
	\alpha \leq_A \alpha' \quad \beta \leq_A \beta'
\justifies
	\alpha \Inter \beta \leq_A \alpha' \Inter \beta'
\endprooftree
\end{array}
\]
\end{definition}
In the literature, the operator $\Inter$ is called {\em intersection}, and $\omega_A$ the {\em universal type}.
Informally, $\Th_A$ is identified with the set of inequalities $\alpha \leq_A \beta$, so that, abusing terminology, we shall also say that $\leq_A$ is a type theory. 
Clearly, the quotient $\Lang_A/_{\leq_A}$ is an inf-semilattice, with $\Inter$ as meet and $\omega_A$ as top
element.

\medskip
\begin{definition}\label{def:filters}
A non empty $F \subseteq \Lang_A$ is a {\em filter} of $\Th_A$ if it is closed  under $\Inter$ and upward closed w.r.t. $\leq_A$; let $\Filt_A$ be the set of filters of $\Th_A$.
The {\em principal filter over} $\alpha$ is the filter $\uparrow_A\alpha = \Set{\beta \in \Lang_A \mid \alpha \leq_A \beta}$. 
\end{definition}

The following proposition is the fundamental fact about intersection types and $\omega$-algebraic lattices:

\begin{proposition}[Representation theorem]\label{prop:filterDomain}
The partial order $(\Filt_A, \subseteq)$ of the filters of an intersection type theory $\Th_A$ is an $\omega$-algebraic lattice, whose compact elements are the
principal filters. Vice versa, any $\omega$-algebraic lattice $A$ is isomorphic the poset $(\Filt_A, \subseteq)$ of filters of some intersection type theory $\Th_A$.
\end{proposition}

\begin{proof} (Sketch).
The first part is established via a direct argument. To the second part, let $\Lang_A = \Set{\alpha_d \mid d \in K(A)}$ (where each $ \alpha_d $ is a new type constant) which is countable by hypothesis, and take
$\Th_A = \Set{\alpha_d \leq_A \alpha_e \mid e \sqsubseteq d}$. By observing that $\alpha_d \Inter_A \alpha_e = \alpha_{d \sqcup e}$, where  
$d \sqcup e \in K(A)$ if $d,e \in K(A)$, and that $\omega_A = \alpha_\top$, we have that $\Th_A$ is an intersection type theory and that
$\Lang_A/_{\leq_A}$ is isomorphic to $K^{\textit{op}} (A)$  ordered by the opposite to the ordering of $A$. From this it follows that $\Filt_A$ is isomorphic
to the ideal completion of $K(A)$, which in turn is isomorphic to $A$ by algebraicity.
\end{proof}

The above proposition substantiates the claim that intersection type theories are the {\em logic} of the domains in $\ALG$.
The theory $\Th_A$ is called the Lindenbaum algebra of $A$ in \cite{Abramsky'91}; Prop. \ref{prop:filterDomain} is the object part of the construction 
establishing that intersection type theories, together with a suitable notion of morphisms among them, do form a category that is equivalent to $\ALG$; such equivalence is
an instance of Stone duality as studied in  \cite{Abramsky'91} w.r.t. the category of 2/3 SFP domains, 
when $\ALG$ is viewed as a subcategory of topological spaces: see e.g. \cite{alesdezahons04}.

The second part in the proof of Prop. \ref{prop:filterDomain} is the most relevant to us: it provides a recipe to describe a domain via a formal system, deriving
inequalities among type expressions that encode ``finite approximations'' of points in a domain, hence of the denotations of terms if the domain is a $\lamImp$-model.
Therefore, we seek theories $\Th_D$ and $\Th_S$ such that:
\begin{equation}\label{eq:filter-equation}
 \Filt_D \iso [\Filt_D \to  [\Filt_S \to (\Filt_D \times \Filt_S)_\bot]]
 \end{equation}

The first step is to show how the functors involved in the equation above, namely the lifting, the product and the (continuous) function space, of which $S = [\Label \to D]$
 is a special case, can be put in correspondence with the construction of new theories out of the theories which determine the domains combined by the functors.
 As this is known after \cite{Abramsky'91}, we just recall their definitions, and fix the notation.
 
 \begin{definition}\label{def:functor-theories}
Suppose that the theories $\Th_A$ and $\Th_B$ are given, then for $\alpha\in \Lang_A$ and $\beta\in \Lang_B$ define:
\[\begin{array}{llll@{\hspace{1cm}}llll}
	\Lang_{A_\bot}~  & \psi & ::= & \alpha \mid \psi \Inter \psi' \mid \omega_{A_\bot} 
	&
	\Lang_{A\times B}~  & \pi & ::= & \alpha \times \beta \mid \pi \Inter \pi' \mid \omega_{A\times B} 
	\\
	\Lang_{A\to B}~  & \phi & ::= & \alpha \to \beta \mid \phi \Inter \phi' \mid \omega_{A\to B} 
	&
	\Lang_{ st A}~  & \sigma & ::= & \tuple{\ell:\alpha} \mid \sigma \Inter \sigma' \mid \omega_{S}
\end{array}\]
Then, we define the following sets of axioms:
\begin{enumerate}
\item $\alpha\leq_{A}\alpha' \Then \alpha \leq_{A_\bot}\alpha'$. 

\item $\omega_{A \times B} \leq_{A\times B} \omega_A \times \omega_B$ and all instances of
	$(\alpha \times \beta) \Inter (\alpha' \times \beta') \leq_{A\times B} (\alpha \Inter \alpha' )\times (\beta \Inter \beta')$.

\item $\omega_{A \to B} \leq_{A\to B} \omega_A \to \omega_B$ and all instances of
	$(\alpha \to \beta) \Inter (\alpha \to \beta') \leq_{A\to B} \alpha \to (\beta \Inter \beta')$.
	
\item $\omega_S \leq_{st A} \tuple{\ell : \omega_A}$ and all instances of 
	$\tuple{\ell : \alpha}\Inter \tuple{\ell : \alpha'}\leq_{st A} \tuple{\ell : \alpha\Inter \alpha'}$.

\end{enumerate}
Finally, the theories $\Th_{A \times B}$, $\Th_{A\to B}$ and $\Th_{st A}$ are respectively closed under the rules:
\[\begin{array}{c@{\hspace{1cm}}c@{\hspace{1cm}}c}
	\prooftree
		\alpha \leq_{A} \alpha' \quad \beta \leq_B \beta'
	\justifies
		\alpha \times \beta \leq_{A\times B} \alpha' \times \beta'
	\endprooftree
&
	\prooftree
		\alpha' \leq_{A} \alpha \quad \beta \leq_B \beta'
	\justifies
		\alpha \to \beta \leq_{A\to B} \alpha' \to \beta'
	\endprooftree
&
	\prooftree
		\alpha\leq_{A}\alpha' 
	\justifies
		\tuple{\ell: \alpha} \leq_{st A} \tuple{\ell: \alpha'}
	\endprooftree

\end{array}
\]
\end{definition}

If $X \in \Filt_{A\to B}$ and $Y \in \Filt_A$, where either $A = D$ and $B = \GS\, D$ or $A = S$ and $B = (D\times S)_\bot$, we define:
\[ X \cdot Y = \Set{\psi \mid \exists \varphi \in Y.\; \varphi \to \psi \in X}\]
We write $\varphi =_A \psi$ if $\varphi \leq_a \psi \leq_A \varphi$. 
For a map $f: \Filt_A \to \Filt_B$, define \[\Lambda(f) = \; \uparrow_{A\to B} \Set{\varphi \to \psi \in \Lang_{A \to B}\mid \psi \in f(\uparrow_A \varphi)}\]

 \begin{lemma}\label{lem:filter-application}
 If $X \in \Filt_{A\to B}$ and $Y \in \Filt_A$ then $X \cdot Y \in \Filt_B$. Moreover the map $\_ \,\cdot\, \_$ is continuous in both its arguments.
 \end{lemma} 
 
 \begin{proof}
 We have $X \ni \omega_{A \to B} \leq \omega_A \to \omega_B$ and that $\omega_A \in Y$, hence $\omega_B \in X \cdot Y$. Since $\to$ is monotonic in its second argument,
 $X \cdot Y$ is upward closed. Finally, if $\psi_1, \psi_2 \in X \cdot Y$ then $\varphi_i \in Y$ and $\varphi_i \to \psi_i \in X$ for $i = 1,2$ and some $\varphi_1,\varphi_2$; then,
 by antimonotonicity of $\to$ w.r.t. its first argument,
 $\varphi_i \to \psi_i \leq \varphi_1 \Inter \varphi_2 \to \psi_i \in X$ being $X$ upward closed, and
 $\varphi_1 \Inter \varphi_2 \to \psi_1 \Inter \psi_2 = (\varphi_1 \Inter \varphi_2 \to \psi_1) \Inter (\varphi_1 \Inter \varphi_2 \to \psi_2) \in X$ since $X$ is closed under intersections.
 
Concerning continuity, we have to show that $X\cdot Y = \bigsqcup {\cal Z}$ where 
${\cal Z} = \Set{\uparrow \varphi \;\cdot \uparrow \psi \mid \varphi \in X \And \psi \in Y}$, and
the sup of a family of filters is the least filter including the union of the family. Inclusion from left to right follows by observing that if $\chi \in X \cdot Y$ then
$\psi \to \chi \in X$ for some $\psi \in Y$, and of course $\chi \in \;\uparrow (\psi \to \chi) \;\cdot \uparrow \psi$. Viceversa if $\chi \in \bigsqcup  {\cal Z}$
then for finitely many $\psi_i,\chi_i$ we have that $\bigcap_i \chi_i \leq \chi$, $\psi_i \in Y$ and $\psi_i \to \chi_i \in X$. It follows that $\chi_i \in X \cdot Y$ for all $i$,
and hence the thesis since $X\cdot Y$ is a filter by the above.
 \end{proof}

 \begin{lemma}\label{abstraction-lemma}
 If $f: \Filt_A \to \Filt_B$ is continuous then $\Lambda(f) \in \Filt_{A \to B}$. $\Lambda(\cdot)$ is itself continuous and such that:
 \[\Lambda(f) \cdot X = f(X)  \qquad \Lambda(\metalambda Y.(X \cdot Y)) = X
 \]
 \end{lemma}
 \begin{proof}
 Easy by unfolding definitions and by Lemma \ref{lem:filter-application}.
 \end{proof}

Now we can establish:

\begin{proposition}\label{prop:functor-filters}
The following are isomorphisms in $\ALG$: 
\[\Filt_{A_\bot} \iso (\Filt_A)_\bot, \quad \Filt_{A \times B} \iso \Filt_A \times \Filt_B, \quad 
\Filt_{A \to B} \iso [\Filt_A \to \Filt_B] \quad \mbox{and} \quad \Filt_{st A} \iso \Label \to \Filt_A.
\] 
\end{proposition}

\begin{proof}
That $\Filt_{A_\bot} \iso (\Filt_A)_\bot$ is a consequence of the fact that $\omega_A <_{A_\bot} \omega_{A_\bot}$ is strict, hence
$\uparrow \omega_{A_\bot}$ is the new bottom added to $\Filt_A$. 
$\Filt_{A \times B} \iso \Filt_A \times \Filt_B$ is induced by the continuous extension of the map $\uparrow (\alpha \times \beta) \mapsto (\uparrow\alpha, \uparrow\beta)$,
that is clearly invertible. Finally, $\Filt_{A \to B} \iso [\Filt_A \to \Filt_B]$ is Lemma \ref{abstraction-lemma}, and $\Filt_{st A} \iso [\Label \to \Filt_A]$ can be reduced to 
a particular case of the latter by using the continuous extension of the map $\uparrow \tuple{\ell: \delta} \mapsto (\ell \mapsto \uparrow \delta)$.
\end{proof}

The next step is to apply Prop. \ref{prop:functor-filters} to describe the compact elements of $D \iso \Filt_D$ and of $S \iso  \Filt_S$ and the (inverse of) their orderings. 
Alas, this cannot be done directly because of the recursive
nature of the equation (\ref{eq:filter-equation}), but it can be obtained by mirroring the inverse limit construction, e.g. along the lines of 
\cite{alesdezahons04,Alessi-Severi'08}. 
Although possible in principle, such a construction requires lots of machinery from the theory of the solution of domain equations; instead we follow 
the shorter path to define the involved type theories below simply by mutual induction: 

\begin{definition}\label{def:typetheories}
Let us abbreviate $S = D^\Label$ and $C = (D\times S)_\bot$; then define the following itt languages by mutual induction:
\[\begin{array}{llll@{\hspace{1cm}}llll}
\Lang_D : & \delta & ::= & \delta \to \tau \mid \delta \Inter \delta' \mid \omegaD & 
\Lang_{\Store} : & \sigma & ::= & \tuple{\ell : \delta} \mid \sigma \Inter \sigma' \mid \omegaS \\
\Lang_{C} : & \kappa & ::= & \delta \times \sigma \mid \kappa \Inter \kappa' \mid \omegaC &  
\Lang_{\GS D} : & \tau & ::= & \sigma \to \kappa \mid \tau \Inter \tau' \mid \omegaSD 
\end{array}\]
We assume that $\Inter$ and $\times$ take precedence over $\to$ and that $\to$ associates to the right so that
$\delta \to \tau \Inter \tau'$ reads as $\delta \to (\tau \Inter \tau')$ and
$ \delta' \to \sigma' \to \delta'' \times \sigma''$ reads as $\delta' \to (\sigma' \to (\delta'' \times \sigma''))$.

Then the relations $\leqD, \leqS, \leqC$ and $\leqSD$ are the least itt's  over the respective languages satisfying the following axioms:
\[\begin{array}{c@{\hspace{1cm}}c}
(\delta \to \tau) \Inter (\delta \to \tau') \leqD \delta \to \tau \Inter \tau'
&
\omegaD \leqD \omegaD \to \omegaSD
\\ [2mm] 
\tuple{\ell : \delta} \Inter \tuple{\ell : \delta'} \leqS \tuple{\ell : \delta \Inter \delta'}
&
\omegaS \leqS \tuple{\ell : \omegaD}
\\ [2mm]
\multicolumn{2}{c}{
(\delta \times \sigma) \Inter (\delta' \times \sigma') \leqC (\delta \Inter \delta') \times (\sigma \Inter \sigma')}
\\ [2mm]
(\sigma \to \kappa) \Inter (\sigma \to \kappa') \leqSD \sigma \to \kappa \Inter \kappa'
&
\omegaSD \leqSD \omegaS \to \omegaC
\end{array}\]
and closed under the rules:
\[\begin{array}{c@{\hspace{8mm}}c@{\hspace{8mm}}c@{\hspace{8mm}}c}
\prooftree
	\delta' \leqD \delta 
	\quad
	\tau \leqSD \tau'
\justifies
	\delta \to \tau \leqD \delta' \to \tau'
\endprooftree
&
\prooftree
	\delta \leqD \delta'
\justifies
	\tuple{\ell : \delta} \leqS \tuple{\ell : \delta'}
\endprooftree
&
\prooftree
	\delta \leqD \delta'
	\quad
	\sigma \leqS \sigma'
\justifies
	\delta \times \sigma \leqC \delta' \times \sigma'
\endprooftree
&
\prooftree
	\sigma' \leqS \sigma 
	\quad
	\kappa \leqC \kappa'
\justifies
	\sigma \to \kappa \leqSD \sigma' \to \kappa'
\endprooftree
\end{array}\]
\end{definition}

\begin{remark}\label{rm:leq-arrow}
	Observe that the only constants used in Def. \ref{def:typetheories} are the $\omega$'s; also we have plenty of equivalences $\varphi = \psi$, namely 
	relations $\varphi \leq \psi \leq \varphi$, involving these constants, that are induced by the definition of the itt's above. For example $\delta \to \omegaSD = \omegaD$ is derivable since
	$\omegaD \leqD \delta \to \omegaSD \leqD \omegaD$ are axioms of $\Th_D$; similarly $\tuple{\ell : \omegaD} = \omegaS$ and $\omegaS \to \omegaC = \omegaSD$. 
	
	However, none of the theories above is trivial, because $\omegaC \not\leqC \omegaD \times \omegaS$. This is implied by
	$\uparrow \omega_C \subset \; \uparrow(\omega_D\times\omega_S)$, that is $\uparrow \omega_C$ corresponds to the new bottom element added to
	$\Filt_{D \times S} \iso \Filt_D \times \Filt_S$ in $\Filt_C = \Filt_{(D\times S)_\bot} \iso (\Filt_{D \times S})_\bot$.
	
	The choice of not having atomic types $\alpha$ in $\Lang_D$ is minimalistic, and parallels the analogous definition 5.2.1 of \cite{Abramski-Ong'93}, where the only 
	constant in the domain logic of a lazy $\lambda$-model $D \iso [D \to D]_\bot$ is $t$ (true), corresponding to our $\omega_D$, and having
	$t \not \leq (t \to t)_\bot$ in the theory.
	
	As a more careful description of $\Filt_D$ would show, by relating its construction to the solution of eq. (\ref{eq:domain-equation}) in Theorem \ref{thm:domain-equation},
	the reason why the $\omega$'s suffice is that the domain equation has a non trivial initial solution. Adding atomic types $\alpha$ to $\Lang_D$ also leads to
	a filter model $\Filt_{D'}$ of $\lamImp$, which is however not isomorphic to $[\Filt_{D'} \to \GS (\Filt_{D'})]$, which is what we show to hold for $\Filt_D$ in the next theorem. 
	To restore the desired isomorphism
	it suffices to add axioms $\alpha =_{D'} \omega_{D'} \to \omega_S \to (\alpha \times \omega_S)$ for all atomic $\alpha$: these correspond to the axioms $\alpha = \omega \to \alpha$
	in \cite{BCD'83}, which are responsible of obtaining a ``natural equated'' solution to the equation $D =[ D \to D]$ of Scott's model: see \cite{Alessi-Severi'08}.
\end{remark}

%

\begin{theorem}\label{thr:iso-D-S2GSD}
The theories $\Th_D$ and $\Th_S$ induce the filter domains $\Filt_D$ and $\Filt_S$ which satisfy equation (\ref{eq:filter-equation}).
\end{theorem}

\begin{proof}
By Prop. \ref{prop:functor-filters}, $\Filt_{D \to \GS D} \iso [\Filt_D \to \Filt_{\GS D}] \iso [\Filt_D \to [\Filt_S \to (\Filt_D \times \Filt_S)_\bot]]$; then 
the thesis follows since $\Filt_D = \Filt_{D \to \GS D}$ by construction.
\end{proof}

\newcommand{\DefEq}{=_{\textit{def}}}

\myparagraph{The $\lamImp$ filter model} 
According to Def. \ref{def:lam-imp-mod}, to show that $\Filt_D$ is a $\lamImp$-model it remains to see that $\Filt_{\GS D} \iso \GS \Filt_D$ can be endowed
with the structure of a monad, which amounts to say that the maps $\UnitF: \Filt_D \to \Filt_{\GS D}$ and 
$\Bind^\Filt: \Filt_{\GS D} \times [\Filt_D \to \Filt_{\GS D}] \to \Filt_{\GS D}$
are definable in such a way to satisfy the monadic laws. This easily follows by instantiating the definition of the state monad Def. \ref{def:statemon} to the subcategory of filter domains:
\begin{equation}\label{eq:UnitF}
\UnitF X \DefEq \Lambda(\metalambda Y \in \Filt_S. (X, Y)) = \; \FILTS{\sigma \to \delta \times \sigma \in \Lang_{\GS D} \mid \delta \in X}
\end{equation}
where ${\tt Filt}\, U$ is the least filter including the set $U$.
On the other hand, observing that $\Filt_{\GS D} \times [\Filt_D \to \Filt_{\GS D}] \to \Filt_{\GS D} \iso \Filt_{\GS D} \times \Filt_D \to \Filt_{\GS D}$, 
for all $X \in \Filt_{\GS D} $, $Y \in \Filt_D$ and $Z \in \Filt_S$ we expect that:
\[(X \Bind^\Filt Y) \cdot Z = \Let{(U, V)}{X\cdot Z}{(Y \cdot U)\cdot V} = \left\{
	\begin{array}{ll}
	(Y \cdot U)\cdot V & \mbox{if $X\cdot Z = U \times V \neq \;\uparrow_C \omega_C$} \\[1mm]
	 \uparrow_C \omega_C & \mbox{otherwise}
	\end{array}\right.
\]
Hence we put
\begin{equation}\label{eq:BindF}
X \Bind^\Filt Y \DefEq  
	\FILTS{\sigma \to \delta'' \times \sigma'' \in \Lang_{\GS D} \mid  
		\exists \delta', \sigma'.\; \sigma \to \delta' \times \sigma' \in X \And
			\delta' \to \sigma' \to \delta'' \times \sigma'' \in Y
	} \; \cup \; \uparrow_C \omega_C
\end{equation}

The final step is to define the interpretations of 
$\textit{get}_\ell : (\Filt_D \to \Filt_{\GS D}) \to \Filt_{\GS D} \iso \Filt_D \to \Filt_{\GS D}$ and 
$\textit{set}_\ell: \Filt_D \times \Filt_{\GS D} \to \Filt_{\GS D}$, which also are derivable from their interpretation into a generic
model $\cal D$. Indeed, according to Def. \ref{def:lambda-imp-model}, for any $X \in \Filt_D$ and $Y \in \Filt_S$ we must have:
\[ \getF{\ell}{X} \cdot Y = X \cdot (Y \cdot \Set{\ell}) \cdot Y\]
where we assume that $\_\,\cdot\,\_$ associates to the left, and we abbreviate $Y \cdot \Set{\ell} = \Set{\delta \in \Lang_D \mid \tuple{\ell:\delta} \in Y}$ representing the
application of the ``store'' $Y$ to the location $\ell$.
Now, let $Z \DefEq  \Set{\tau \in \Lang_{\GS D} \mid \exists \delta \in Y\cdot \Set{\ell}.\; \delta\to\tau \in X}$ then
\[ X \cdot (Y \cdot \Set{\ell}) \cdot Y 
= Z \cdot Y\]
If $\tau = \omega_{\GS D}$ then $\delta\to\tau = \omega_D$, which trivially belongs to any filter; if instead $\tau \neq \omega_{\GS D}$ then 
$\tau = \bigwedge_I \sigma_i \to\kappa_i$ and $\delta\to\tau \in X$ if and only if $\delta\to \sigma_i \to\kappa_i \in X$ for all $i\in I$. From this we conclude that
\[ Z \cdot Y = \Set{\kappa \in \Lang_C \mid \exists \sigma \in Y \mid \sigma \to \kappa \in Z} = 
	\Set{\kappa \in \Lang_C \mid \exists \tuple{\ell:\delta} \Inter \sigma \in Y.\; \delta \to \sigma \to \kappa \in X}\]
and therefore the appropriate definition of $\getF{\ell}{X}$ is
\begin{equation}\label{eq:getF}
\getF{\ell}{X} \DefEq \;\FILTS{(\tuple{\ell:\delta} \Inter \sigma) \to \kappa \in \Lang_{\GS D} \mid \delta \to (\sigma \to \kappa)\in X} 
\end{equation}
Similarly, again by Def. \ref{def:lambda-imp-model}, for any $X \in \Filt_D$,  $Y \in \Filt_{\GS D}$ and $Z \in \Filt_S$ we expect:
\[  \setF{\ell}{X}{Y} \cdot Z = Y \cdot (Z[\ell \mapsto X])\]
where $Z[\ell \mapsto X]$ is supposed to represent the update of $Z$ by associating $X$ to $\ell$, namely:
\[ Z[\ell \mapsto X] = \FILT{\Set{\tuple{\ell : \delta} \mid \delta \in X }\cup \Set{\tuple{\ell' : \delta' }\mid \tuple{\ell' : \delta' }\in Z \;\And\; \ell' \neq \ell}} \]
Then
\[ \begin{array}{rcl}
Y \cdot (Z[\ell \mapsto X]) & = & \FILTS{\kappa \mid \exists \sigma \to \kappa \in Z[\ell \mapsto X].\; \sigma \to \kappa \in Y} \\ [1mm]
& = & \FILTS{\kappa \mid \exists \sigma'  \in Z, \delta \in X.\; \tuple{\ell: \delta} \Inter \sigma' \to \kappa \in Y\;\And\; \ell \not \in \dom{\sigma'}}
\end{array}\]
and therefore we define:
\begin{equation}\label{eq:setF}
\setF{\ell}{X}{Y} \DefEq \FILTS{\sigma' \to \kappa \in \Lang_{\GS D} \mid \exists \delta \in X.\; \tuple{\ell: \delta} \Inter \sigma' \to \kappa \in Y\;\And\; \ell \not \in \dom{\sigma'}}
\end{equation}
Eventually, by construction we have:

\begin{theorem}\label{thr:filter-model}
The structure $\Filt = (\Filt_D, \GS, \Sem{\cdot}^{\Filt_D},  \Sem{\cdot}^{\Filt_{\GS D}}  )$ is a $\lamImp$-model where, for $e: \Var \to \Filt_D$ the interpretations
$\Sem{V}^{\Filt_D} e \in \Filt_D$ and $\Sem{M}^{\Filt_{\GS D}} e \in \Filt_{\GS D}$ are defined by:
\[\begin{array}{rcl@{\hspace{2cm}}rcl}
\Sem{x}^{\Filt_D} e & = & e(x) & \Sem{\lambda x.M}^{\Filt_D} e & = & \Lambda(\metalambda X \in \Filt_D.\; \Sem{M}^{\Filt_{\GS D}} e[x\mapsto X]) \\[1mm]
\Sem{[V]}^{\Filt_{\GS D}} e & = & \UnitF (\Sem{V}^{\Filt_D} e) & \Sem{M \Bind V}^{\Filt_{\GS D}} e & = & (\Sem{M}^{\Filt_{\GS D}} e) \Bind^\Filt (\Sem{V}^{\Filt_D} e) \\ [1mm]
\Sem{\get{\ell}{\lambda x.M}}^{\Filt_{\GS D}} e & = & \getF{\ell}{\Sem{\lambda x.M}^{\Filt_D} e} & 
	\Sem{\set{\ell}{V}{M}}^{\Filt_{\GS D}} e & = & \setF{\ell}{\Sem{V}^{\Filt_D} }{\Sem{M}^{\Filt_{\GS D}} e }
\end{array}
\]
\end{theorem}

\begin{proof} By applying equations (\ref{eq:UnitF})-(\ref{eq:setF}) to the clauses of Def. \ref{def:lambda-imp-model}.
\end{proof}

\section{Deriving the type assignment system}\label{sec:deriving}

The definition of a type assignment system can be obtained out of the construction of the filter-model.
 The system is nothing else than the description of the denotation of terms in the particular $\lamImp$-model $\Filt$, which is possible because both $\Sem{V}^{\Filt_D} e \in \Filt_D$ and $\Sem{M}^{\Filt_{\GS D}} e \in \Filt_{\GS D}$, and filters are sets of types. 

Types are naturally interpreted as subsets of the domains of term interpretation, namely $\Filt_D$ or $\Filt_{\GS D}$ according to their kind; by applying to the present case the same construction as in \cite{BCD'83} and \cite{Abramsky'91}, which originates from Stone duality for boolean algebras, we set:

\begin{definition}\label{def:type-interpretation}
For $A = D, S, C$ and $\GS D$ and $\varphi \in \Lang_A$ define:
\[ \Sem{\varphi}^{\Filt} = \Set{X \in \Filt_A \mid \varphi \in X} \]
\end{definition} 

Such interpretation is a generalization of  the natural interpretation of intersection types over an extended type structure \cite{Coppo-et.al'84}, which in the present case yields:

\begin{proposition}\label{prop:type-interp}
For $A = D, S, C$ and $\GS D$ and $\varphi, \psi \in \Lang_A$ we have:
\[ \Sem{\varphi \Inter \psi}^{\Filt} = \Sem{\varphi}^{\Filt} \cap \Sem{\psi}^{\Filt}, 
\qquad \Sem{\omega_A}^{\Filt} = \Filt_A \quad \mbox{and} \quad
\varphi \leq \psi \Then \Sem{\varphi}^{\Filt} \subseteq \Sem{\psi}^{\Filt} \]
Moreover:
\begin{enumerate}
\item \label{prop:type-interp-i}
	$\Sem{\delta \to \tau}^{\Filt} = \Set{X \in \Filt_D \mid \forall Y \in \Sem{\delta}^{\Filt}.\; X \cdot Y \in \Sem{\tau}^{\Filt} }$
\item $\Sem{\tuple{\ell:\delta}}^{\Filt} = \Set{X \in \Filt_S \mid X\cdot \Set{\ell} \in \Sem{\delta}^{\Filt}}$
\item $\Sem{\delta \times \sigma}^{\Filt} = \Set{X \in \Filt_C \mid \pi_1(X) \in \Sem{\delta}^{\Filt} \And \pi_2(X) \in \Sem{\sigma}^{\Filt} } $
\item $\Sem{\sigma \to \kappa}^{\Filt} = \Set{X \in \Filt_{\GS D} \mid \forall Y \in \Sem{\sigma}^{\Filt}. \; X \cdot Y \in \Sem{\kappa}^{\Filt}  }$
\end{enumerate}
where, for $X \in \Filt_C \iso (\Filt_D \times \Filt_S)_\bot$,  $\pi_1(X) = \Set{\delta \in \Lang_D\mid \exists \sigma \in \Lang_S.\; \delta \times \sigma \in X}$ and similarly for $\pi_2(X)$.
\end{proposition}

\begin{proof} The first part is immediate from the definition of type interpretation and of filters. Concerning the remaining statements,
we prove just (\ref{prop:type-interp-i}) as the other ones are similar.
First observe that $\uparrow\varphi \in \Sem{\varphi}^{\Filt}$. Hence, 
if $X \in \Sem{\delta \to \tau}^{\Filt}$ then $\delta \to \tau \in X$, which implies that $X \cdot \uparrow \delta = \;\uparrow \tau \in \Sem{\tau}^{\Filt}$. Viceversa if
$X \cdot Y \in \Sem{\tau}^{\Filt} $ for all $Y \in \Sem{\delta}^{\Filt}$ then in particular $X \cdot \uparrow \delta = \;\uparrow \tau \in \Sem{\tau}^{\Filt}$, which implies that
for some $\delta' \in \; \uparrow \delta$, $\delta'\to \tau \in X$; but then $\delta \leq_D \delta'$ so that $\delta'\to \tau \leq_D \delta\to \tau$ and therefore $\delta \to \tau \in X$
as $X$ is upward closed.
\end{proof}

We can now build the type assignment system as the description, via type derivations, of the semantics of terms and types in the model $\Filt$. First, type contexts 
$\Gamma = \Set{x_1:\delta_1, \ldots, x_n:\delta_n}$ are put into correspondence with environments $e_\Gamma: \Var \to \Filt_D$ by setting
$e_\Gamma(x_i) =\;\uparrow_D \delta_i$. Second, typing judgments translate the statements:
\[ \Gamma \der V:\delta \iff (\Sem{V}^{\Filt_D} e_\Gamma) \in \Sem{\delta}^\Filt \iff \delta \in \Sem{V}^{\Filt_D} e_\Gamma 
\]
and similarly, $\Gamma \der M:\tau \iff \tau \in \Sem{M}^{\Filt_{\GS D}} e_\Gamma$.

Third and final step, typing rules of the shape
\[
\prooftree
	\Gamma_1 \der P_1:\varphi_1 \quad \cdots \quad \Gamma_n \der P_n:\varphi_n
\justifies
	\Gamma \der Q:\psi
\endprooftree
\]
are the translations of equations
\begin{equation}\label{eq:rule-translation}  
\Sem{Q}^\Filt e_\Gamma = \FILTS{\psi \mid \varphi_1 \in \Sem{P_1}^\Filt e_{\Gamma_1} \And \cdots \And \varphi_n \in \Sem{P_n}^\Filt e_{\Gamma_n}}
\end{equation}
derived from Prop. \ref{prop:type-interp} and Thr. \ref{thr:filter-model}. We take advantage of the factorization of the right hand sides of equations 
(\ref{eq:rule-translation}) into a set of types $U$ and
its closure $\texttt{Filt}(U)$, by adding the rules:
\begin{equation}\label{eq:filter-rules}
	\begin{array}{c@{\hspace{1cm}}c@{\hspace{1cm}}c}
		\prooftree
		\vspace{3mm}
		\justifies
		\Gamma \der P : \omega
		\using \omegaR
		\endprooftree
		&
		\prooftree
		\Gamma \der P : \varphi
		\quad
		\Gamma \der P : \psi
		\justifies
		\Gamma \der P : \varphi \Inter \psi
		\using \interR
		\endprooftree
		&
		\prooftree
		\Gamma \der P : \varphi
		\quad
		\varphi \leq \psi
		\justifies
		\Gamma \der P : \psi
		\using \leqR
		\endprooftree
	\end{array}
\end{equation}
Therefore we are left to translate set inclusions 

\begin{equation}\label{eq:rules}
	\Sem{Q}^\Filt e_\Gamma \supseteq \Set{\psi \mid \varphi_1 \in \Sem{P_1}^\Filt e_{\Gamma_1} \And \cdots \And \varphi_n \in \Sem{P_n}^\Filt e_{\Gamma_n}}
\end{equation}
where
$\psi$ is neither an intersection nor an $\omega$:

\begin{definition}\label{def:intersTypeSys}
({\bf Intersection type assignement system for $\lamImp$})
The system is obtained by adding to the rules (\ref{eq:filter-rules}) the following:
\[\begin{array}{c@{\hspace{0.6cm}}c}
		\prooftree
		\vspace{3mm}
		\justifies
		\Gamma, x:\delta \der x : \delta
		\using \varR
		\endprooftree
		&
		\prooftree
		\Gamma, \, x:\delta \der M : \tau
		\justifies
		\Gamma \der \lambda x.M : \delta \to \tau
		\using \lambdaR
		\endprooftree
		\\ [6mm]
		\prooftree
		\Gamma \der V:\delta
		\justifies
		\Gamma \der \unit{V}: \sigma \to \delta \times \sigma
		\using \unitR
		\endprooftree
		&
		\prooftree
		\Gamma \der M : \sigma \to \delta' \times \sigma'
		\quad
		\Gamma \der V : \delta' \to \sigma' \to \delta'' \times \sigma''
		\justifies
		\Gamma \der M \Bind V : \sigma \to \delta'' \times \sigma''
		\using \bindR
		\endprooftree
		\\ [6mm]
		\prooftree
		\Gamma, \, x:\delta \der M : \sigma \to \kappa
		\justifies
		\Gamma \der \get{\ell}{\lambda x.M} : (\tuple{\ell : \delta} \Inter \sigma) \to \kappa
		\using \getR
		\endprooftree
		&
		\prooftree
		\Gamma \der V : \delta
		\quad
		\Gamma \der M :  (\tuple{\ell : \delta} \Inter \sigma)  \to \kappa
		\quad
		\ell \not\in \dom{\sigma} 
		\justifies
		\Gamma \der \set{\ell}{V}{M} : \sigma \to \kappa
		\using \setR
		\endprooftree
	\end{array}
	\]

\end{definition}

The rules in Definition \ref{def:intersTypeSys} are obtained by instantiating the inclusion (\ref{eq:rules}) to equations (\ref{eq:UnitF}) - (\ref{eq:setF}). In conclusion, we obtain for our type system the analogous of the ``Type-semantics theorem'' for
intersection types and the ordinary $\lambda$-calculus (see \cite{BarendregtDS2013}, Thr. 16.2.7):

\begin{theorem}\label{thr:type-semantics}
Let $e\models^\Filt \Gamma$ if and only if  $e(x) \in \Sem{\delta}^\Filt$ whenever $x:\delta \in \Gamma$; then
\begin{enumerate}
		\item $\Sem{V}^{\Filt_D}e = 
		\Set{\delta \in \Lang_D \mid \exists \Gamma. \; e \models^\Filt  \Gamma \And \Gamma \der V:\delta}$
		\item $\Sem{M}^{\Filt_{\GS D}}e = 
		\Set{\tau \in \Lang_{\GS D} \mid \exists \Gamma. \; e \models^\Filt  \Gamma \And \Gamma \der M:\tau}$
\end{enumerate}
\end{theorem}

\begin{proof}
Observe that $e\models^\Filt \Gamma$ if and only if $e_\Gamma(x) = \; \uparrow\delta \subseteq e(x)$ if $x:\delta \in \Gamma$; since any filter $X$ is
the union of the principal filters $\uparrow \delta$ with $\delta \in X$, the proof reduces to establishing that
$\Sem{V}^{\Filt_D}e_\Gamma = \Set{\delta \in \Lang_D \mid  \Gamma \der V:\delta}$, and similarly for 
$\Sem{M}^{\Filt_{\GS D}}e_\Gamma$; the latter are an easy induction over term interpretations into $\Filt$ for left to right inclusions, and over type derivations
for the inclusions from right to left.
\end{proof}

In spite of the complexity of the construction we have been going through in the last sections, the payoff of our work is a system with just one typing rule for
each syntactical construct in the grammar of the calculus, plus the ``logical'' rules (\ref{eq:filter-rules}) which are standard in intersection type systems with subtyping since
\cite{BCD'83}. 

Moreover, by the obvious interpretation $\Sem{\varphi}^{\cal D}$ of types in a model $\cal D$, of which 
Prop.  \ref{prop:type-interp} is the instance in case of $\Filt$, we get soundness and completeness theorems for free. Indeed, without entering into the details, we claim that:
\[ \Gamma \models V:\delta  ~ \iff ~  \Gamma \der V :\delta \quad \mbox{and} \quad
	\Gamma \models M:\tau ~ \iff ~ \Gamma \der M :\tau. 
\]
where for a $\lamImp$-model ${\cal D}$ we write $\Gamma \models^{\cal D} T : \varphi$ if $\Sem{T}^{\cal D} e \in \Sem{\varphi}^{\cal D}$ for all $e\models^{\cal D} \Gamma$,
and write $\Gamma \models T : \varphi$ if  $\Gamma \models^{\cal D} T : \varphi$ for all $\cal D$. 



\section{Discussion and Related works}\label{sec:related}

This work is a development of \cite{deLiguoroTreglia20}, where we considered a pure untyped computational $\lambda$-calculus, namely without operations nor constants.
Therefore, the monad $T$ and the respective unit and bind were generic, so that types cannot convey any specific information about the domain $TD$, nor about effects represented by the monad. 

In the unpublished report \cite{deLiguoroTreglia2021} we sketched the construction of the type system for the state monad $\GS$. This has been motivated on the
ground of an operational interpretation of terms that we illustrated, culminating in the characterization of convergence via the type system.
This is akin to the lazy $\lambda$-calculus \cite{Abramski-Ong'93}, where convergent terms are characterized  by the formula $(t \to t)_\bot$ of the endogenous logic of the calculus, in the sense of \cite{Abramsky'91}. 
 
While deferring the exposition of the convergence characterization to a companion paper, here we have focussed on the construction of the type system itself. 
Since \cite{BCD'83} we know that a $\lambda$-model can be constructed by taking the filters of types in an intersection type system with subtyping. The relation
among the filter-model and Scott's $D_\infty$ construction has been subject to extensive studies, starting with \cite{Coppo-et.al'84} and continuing with 
\cite{Dezani-CiancagliniHA03,alesdezahons04,Alessi-Barbanera-Dezani'06,Alessi-Severi'08}. In the meantime Abramsky's theory of domain logic in \cite{Abramsky'91}
provided a generalization of the same ideas to algebraic domains in the category of 2/3 SFP, of which $\ALG$ is a (full) subcategory, based on Stone duality.

In the present work, we use the mathematical theories just mentioned, but reversing the process from semantics to types. Usually, one starts with a calculus
and a type system, looking for a semantics and studying properties of the model. On the contrary, we move from a domain equation and the definition of the denotational
semantics of the calculus of interest and synthesize a filter-model and an intersection type system.
This synthetic use of domain logic is, in our view, prototypical w.r.t. the construction
of type systems catching properties of any computational $\lambda$-calculus with operators.
We expect that the study of such systems will be of help in understanding the operational semantics of such models, a topic that has been addressed in \cite{LagoGL17,LagoG19}, but with different mathematical tools.

A further research direction is to move from $\ALG$ to other categories such as the category of relations. The latter is deeply related to non-idempotent intersection types \cite{Carvalho18,BucciarelliEM07} that have been shown to catch intensional aspects of $\lambda$-calculi involving quantitative reasoning about computational resources. 
If the present approach can be rephrased in the category of relations, then our method could produce non-idempotent intersection type systems for effectful $\lambda$-calculi. 

In the perspective of considering categories other than $\ALG$, it is natural to ask whether the synthesis of an ``estrinisic'' type system in the sense of Reynolds out of the denotational semantics of a calculus, either typed or not, can be described in categorical terms. Intersection type theories and related systems, as we have been using in this work, are a special case of refinement type systems;
a starting point may be the framing of refinement type systems as functors in \cite{MelliesZeilberger2015}.

\myparagraph{Acknowledgements} We are grateful to the anonymous referees for their comments and hints to categorical prospectives of the present work.

\bibliographystyle{eptcs}
\bibliography{references}
\end{document}